\documentclass{aa}
  
\usepackage[varg]{txfonts}
\usepackage{graphicx}
\usepackage[dvipsnames]{xcolor}
\usepackage{color,epsfig,amsmath,bm}
\usepackage{ulem}

\usepackage{amsthm}

\begin{document}

\title{The existence of null points of flow and magnetic fields as prerequisites for the presence of astropauses and their co-location}


 \author{Dieter H. Nickeler\inst{1} \corrauth{dieter.nickeler@asu.cas.cz}
 \and Kuljeet Singh Saddal\inst{1,2}\email{kuljeet.saddal@gmail.com}
 \and Rodrigo Meneses Pacheco\inst{3}\email{rodrigo.meneses@uv.cl}
 \and Michaela Kraus\inst{1}\email{michaela.kraus@asu.cas.cz}
        }

   \institute{Astronomical Institute, Czech Academy of Sciences, Fri\v{c}ova 298, 251 65 Ond\v{r}ejov, Czech Republic
   \and Astronomical Institute, Faculty of Mathematics and Physics, Charles University, V Hole\v{s}ovick\'{a}ch 2, 182 00 Prague, Czech Republic
   \and Escuela de Ingenier\'{i}a Civil, Universidad de Valpara\'{i}so, General Cruz 222, Valpara\'{i}so, Chile}

   \date{Received xxx; accepted xxx}

\authorrunning{Nickeler et al.}
\titlerunning{The existence and co-location of null points}

 
  \abstract
   {Astropauses can be defined via magnetic or fluid separatrices. For the boundary layer between the stellar wind and the local interstellar medium, two separatrices can occur, the magnetopause and the hydropause. The question is whether they are identical.} 
   {To address this, we investigate under which conditions the null points of the magnetic and velocity field, which define the corresponding separatrices, coincide spatially.}
   {We analysed the case of a non-monotonous pressure distribution driving stationary counterstreaming ideal hydrodynamic flows and demonstrate that the existence of an extremum of the pressure in the three-dimensional case together with a non-degeneracy condition on the velocity field implies that this point is also a stagnation point. Turning to ideal magnetohydrodynamics (MHD) and considering the regularity of all fields, we formulated the momentum equations, revealing their quasi-linearity. This allowed us to prove the occurrence of null points by examining the rank of the Jacobians of the vector fields.} 
   {We found that for a given null point of either the velocity or the magnetic field at the pressure extremum, the pendant vector field has a co-located null if the Jacobian of this field is non-singular. We extended the analysis to resistive MHD and further to MHD with general non-ideal terms and proved that if at least one of the first-order momentum equations is ideal, the stagnation point and the magnetic null point are co-located. Even for general fluid theory, the two null points co-locate as long as the curl of the generalised non-ideal terms is zero at one of them or can be represented by a linear operator in the magnetic and velocity field.}
   {In  general, it is not excluded that the velocity and magnetic nulls are dislocated. However, this would require specific properties of the curl of the non-ideal term, which would reduce the probability of a relaxed quasi-stationary configuration.}


   \maketitle
   \nolinenumbers

\section{Introduction}

Large-scale flows and magnetic fields are ubiquitous in cosmic plasmas. When plasma flows from 
different sources encounter each other, such as the interstellar medium (ISM) and stellar outflows or winds, 
complex interfaces can form. The plasmas of the individual sources within the framework of macroscopic physics cannot penetrate each other freely but have to interact (magneto-) hydrodynamically. 
At the interface, where the interaction takes place, displacement and shearing motions of the involved fields can induce strong currents or vortices, called current or vortex sheets. Within these current or vortex sheets, isolated null points of the corresponding vector fields can form and persist, such as magnetic null (or neutral) points and velocity null (or stagnation) points, which are the navels or vertices of the magnetopauses and hydropauses.

At the boundary between the ISM and the stellar wind flow, a contact surface is formed that is called 
the astropause (or heliopause, in the case of the Sun). It is usually referred to as the hydropause
\citep{1961ApJ...134...20P, 1972NASSP.308..609A, 1991SSRv...58..193F, 1998JGR...103.1889L}.
In the specific case of field-aligned flows with a regular Alfv\'{e}n Mach number, the hydropause and magnetopause are identical \citep[e.g.][]{2006A&A...454..797N}. However,
an as yet unanswered question, posed and discussed by \citet{2008ASTRA...4....7N} and \citet{2014ASTRP...1...51N} and revisited by \citet{2022SSRv..218...36K} and \citet{2023FrASS..1060618L}, is which physical value defines this contact surface in general, the magnetic or the velocity field.

Within the framework of a quasi-stationary approach, the plasma on either side of the astropause typically satisfies different conditions, but across the boundary, the pressure must have developed an equilibrium. 
A general ansatz for a potential flow model with a stagnation point, in which the equilibrium condition across the astropause boundary layer is fulfilled automatically, was introduced by \citet{1993A&A...277..249F}. As in the Parker model \citep{1961ApJ...134...20P}, the stagnation point occurs because the superposition principle enables the homogeneous ISM flow 
(asymptotic boundary condition) to compensate for the radially outwards decreasing monopole outflow (solar wind).

Depending on the velocities of the involved plasma flows and their corresponding phase speeds, shocks can form at which the flows are decelerated from supersonic to subsonic speeds. For the solar (or in general stellar) wind, this occurs at the termination shock, whereas for the ISM flow, this occurs at the bow shock. Therefore, up to three interfaces can form \citep{1976Ap&SS..41..481B, 1993JGR....9815157B}. A sketch of this classical picture \citep[also sometimes called the \lq Baranov\rq~interface, see][]{2000Ap&SS.274...35F} is shown in 
Fig.\,\ref{fig:astrosphere}.

The astrospheric scenario of a multiple surface interface between the interstellar and stellar wind
is a motivation for us to focus on the surface of the astropause boundary layer. In an attempt to develop a kinematic model for the magnetic field in the vicinity of the astropause, \citet{2015ApJ...805..173R} used the ideal subsonic Parker flow, which is a point source superimposed by a homogeneous flow from infinity, and solved the ideal magnetohydrodynamic (MHD) induction
equation. The Parker flow is axis-symmetric, which can inherently lead to diverging magnetic fields on the astropause in case the magnetic field is transported passively by a flow. 
For instance, \citet{1996ApJ...460..185C} found that the introduction of a symmetry of at least one of the involved fields results in non-regular solutions, which seems to be a natural outcome of the 
assumption of a symmetry. This type of non-regular solutions, as found by \citet{2015ApJ...805..173R}, shows a flux pile-up regime so that the energy needs to be artificially dissipated or the field has to be modified \citep{2024ApJ...961..244F} for the infinite field strength at the stagnation point to vanish.

\begin{figure}
\centering
       \includegraphics[width=\columnwidth]{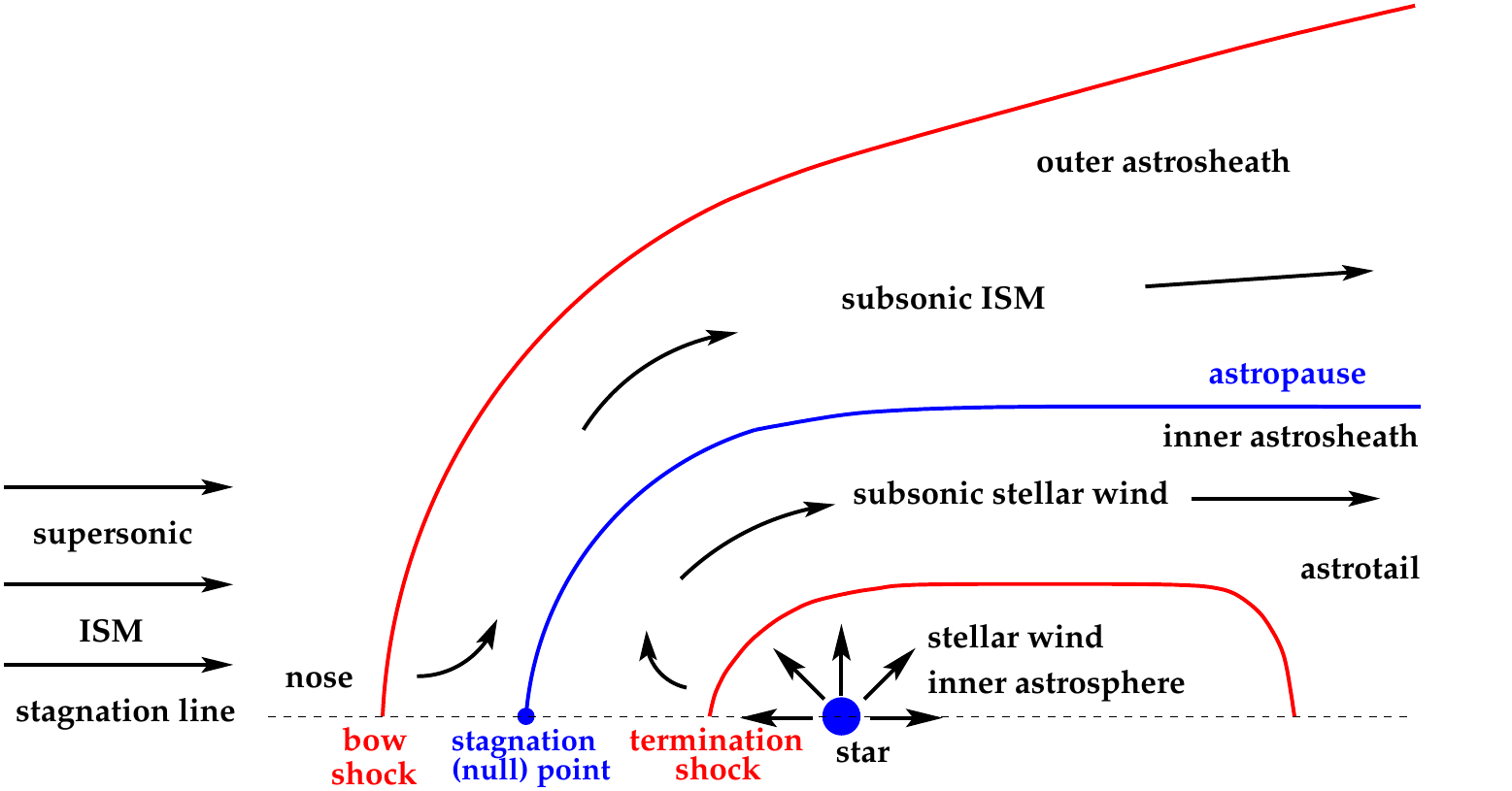}
      \caption{Classical \lq Baranov\rq~interface with three interfaces at most: the bow shock, the astropause, and the termination shock.}
      \label{fig:astrosphere}
      \end{figure}

Our goal is to investigate this counterflow scenario from the perspective of regularity. Regular 
fields have the advantage that no singularities occur. The regularity and structural stability 
of flow configurations with magnetic and flow fields are important for the definition 
of quasi-steady separating surfaces between two flows. Separatrices or separatrix surfaces must appear where topologically disjoint field lines and/or streamlines occur, making it possible to distinguish between separated regions of plasmas and their corresponding dynamics. 
Specific separatrix surfaces that form are, for example, when stars with their expanding winds impinge on the ISM flow \citep[astropause,][]{1961ApJ...134...20P} or when stellar winds encounter the magnetospheres of planets \citep[magnetopause,][]{1963Sci...141..130H}. Magnetic reconnection sites, for example helmet streamers in the solar corona \citep{1998SoPh..180..439W, 2000SoPh..191..391W}, eruptive flare--domains in stellar coronae, often called magnetic domes or fans \citep{2008ApJ...675.1656P}, or flux transfer events in planetary magnetospheres \citep{1978SSRv...22..681R}, are also configurations that depend on typical, \lq generic\rq~or at least structurally stable fields, developing magnetic separatrices. We do not aim to describe reconnection processes but focus on relaxed, stationary flows and fields. This approach is justified because not every non-ideal or resistive dynamics automatically leads to magnetic reconnection \citep[see, e.g.,][]{1998JPlPh..59..303T, 2000JPlPh..64..601T, 2000PhPl....7.3542T, 2014A&A...569A..44N}.

Global topological and geometrical investigations that describe the (tube-like or tail-like) shapes of astrospheres have been performed, for example by \citet{2022A&A...666A.197C}.  
Our investigation focuses primarily 
on the local scenario in the vicinity of the apex of the astrosphere and on the question under which conditions the velocity and magnetic null co-locate. While \citet{2008ASTRA...4....7N} proposed co-location of the 
two nulls in ideal and quasi-ideal\footnote{This means that the non-idealness can be written at least locally as a gradient (e.g. as the gradient of the electron pressure).} MHD using only Ohm's law, the current work proves this finding in a solid way by including the momentum equation as well. Furthermore, we 
extend the analysis by \citet{2008ASTRA...4....7N} to general non-ideal and resistive 3D MHD with the aim to define constraints under which the stagnation point and magnetic null points (and hence, the separatrices) are identical at least locally.

The paper is structured as follows. In Sect.~\ref{sect:joint} we provide an overview of the macro-physical and topological needs for defining the conditions under which a stellar wind and the surrounding ISM stream are separated. We also start  
with the study under which conditions stagnation points in hydrodynamical flows
and null points of the pressure gradient must co-locate.
We extend the analysis to MHD in Sect.~\ref{sect:coloc}, where we prove 
that for a vanishing gradient of the total pressure and a co-located either magnetic null point or stagnation point, the two vector fields must have a null point at the same location when the fields are 3D (Sect.~\ref{sect:theorem}). We further present the constraints for the co-location and dislocation of the nulls for resistive MHD in Sects.~\ref{sect:resistive} and \ref{sect:disloc}, respectively. We briefly discuss Hall MHD in Sect.~\ref{sect:Hall} and the case that the non-ideal term can be written as a quasi-linear operator of the velocity and magnetic field (Sect.~\ref{sect:general}). We provide a summary of the conditions found in Sect.~\ref{sect:disc} and discuss their relevance. We also caution about the mismatch of null and stagnation points in (quasi-)ideal MHD simulations.

\section{Joint or disjoint -- the need for separatrix surfaces for the existence of \lq pauses\rq}
\label{sect:joint}

\subsection{Disjoining surfaces or shock surfaces?}

The existence of null points of vector fields is
required for the spanning of separating surfaces. These surfaces 
form a boundary so that
the topologically disjoint field lines of the corresponding vector fields exist on each
side of the separatrix. To analyse the structure of fields with null points,
we focused on the stationary approximation. 

In fluid dynamics, there are two types of interfaces: shocks and separatrices. However, there is a significant difference between them:
shocks allow for the penetration of plasma flows across the shock surface into regions of identical stream
and magnetic connectivity and vice versa, while a
separatrix guarantees that there is a boundary that separates field lines
(e.g. magnetic field lines or streamlines) into regions of different
connectivity.

In addition to the topological perspective, the question arises which other physical
constraint can require the existence of null points. One possibility is the
course of the pressure, driving stationary counterstreaming MHD flows, for example interstellar
medium flows and the outer stellar wind flows. This pressure cannot be generally assumed to be purely monotonous.
For non-monotonous pressure distributions, there must exist extrema\footnote{In the case of a flow around an obstacle, the matter can be decelerated, causing a stagnation of the flow and an accumulation of matter in front of the obstacle.
This stagnation therefore creates an extremum of the mass density, and the gradient of the density 
vanishes, implying that in barotropic flows, the pressure
gradient vanishes as well.
The case of a saddle point of the pressure cannot be completely excluded. 
However, if a 
conservation law such as the theorem of Bernoulli's principle applies, such as the
barotropic case for typical barotropic laws, the relation between pressure and density is
monotonic and regular. Even for other physically relevant cases, for instance in the
limit of incompressible dynamics, the total pressure must be constant along
streamlines, and therefore, if the kinetic pressure is zero, the thermal pressure
must be at maximum.}, such that the
pressure gradient is zero at some points within a specific domain. For a purely
ideal hydrodynamical (HD) problem, we show in the following that the existence of an
extremum (usually a local maximum) of the thermal or plasma pressure at a
certain point in the generic three-dimensional case automatically implies that
this point is also a stagnation point.  
To prove  this statement, we demand the regularity of all involved fields and their derivatives.

\subsection{The existence of stagnation points}
\label{sect:HD}

      \begin{figure}
      \centering
       \includegraphics[width=0.49\columnwidth]{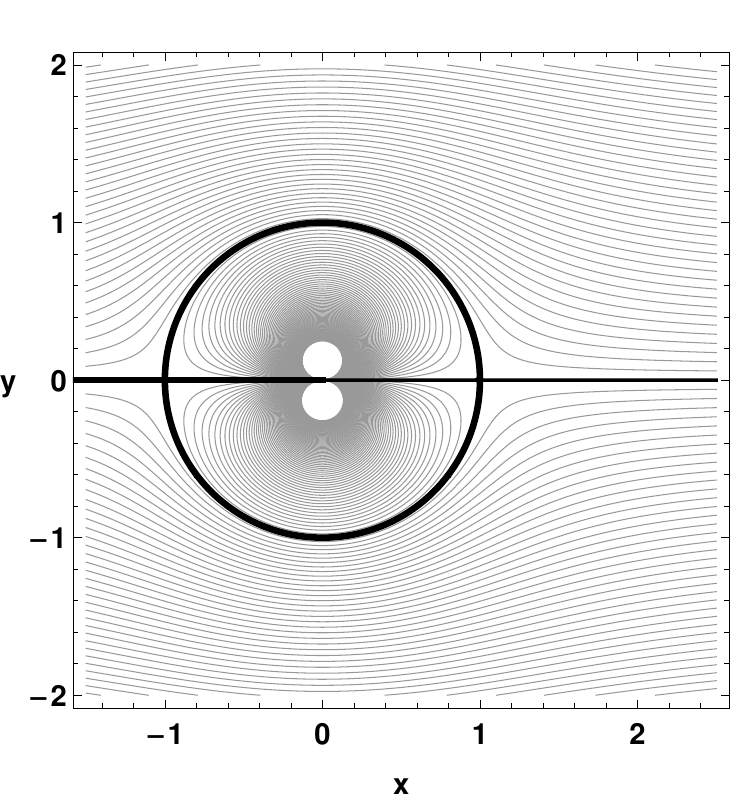}
       \includegraphics[width=0.49\columnwidth]{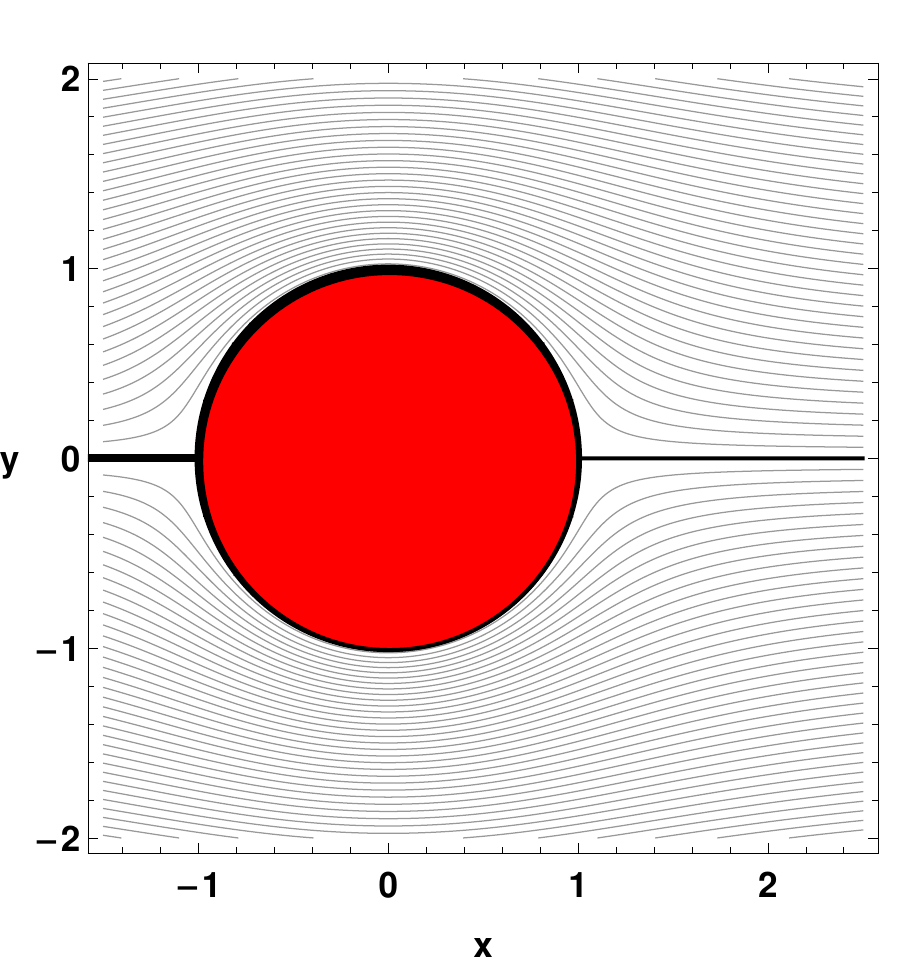}\\
       \includegraphics[width=0.49\columnwidth]{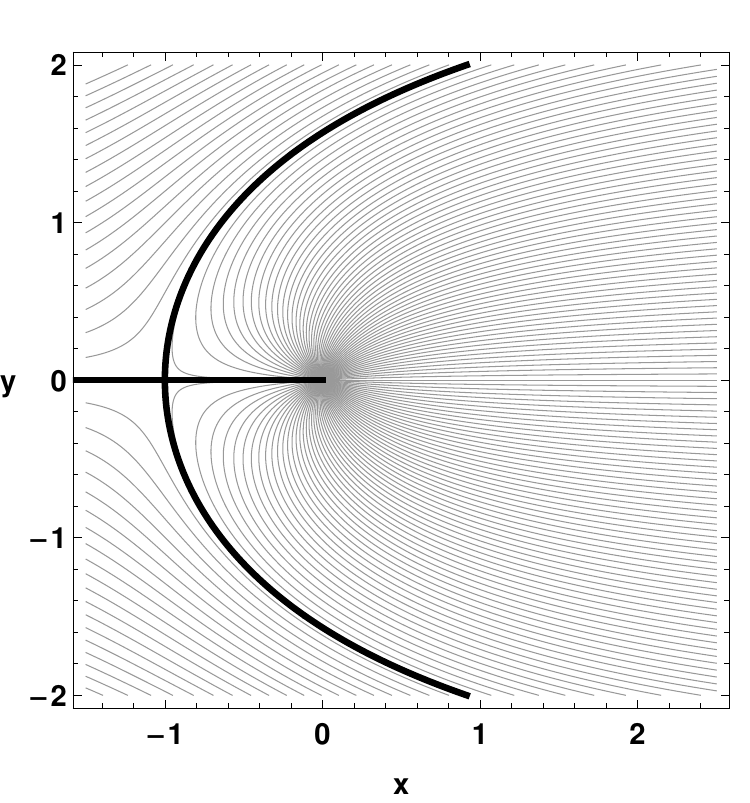}
       \includegraphics[width=0.49\columnwidth]{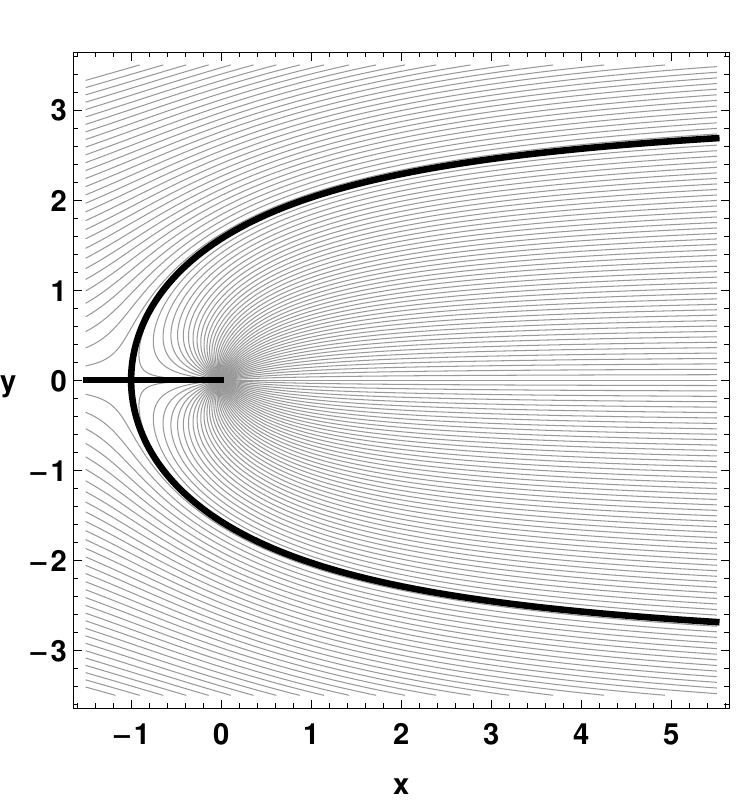}
       \caption{2D flows around obstacles with stagnation lines and separatrix curves (thick solid lines). In all panels, a homogeneous flow approaches from the left and streams around the obstacle. The obstacles are a dipole flow (\textit{top left}) that is substituted by a solid circular body (\textit{top right}). These are a so-called
       \lq Rankine body\rq \ or \lq finite body\rq. Two stagnation points exist that are located at $(x,y) = (\pm1,0)$. In the bottom panels, the obstacle is given by a monopole flow, resulting in a Parker-like tail, a so-called \lq Rankine half body\rq \ or \lq infinite body\rq. In this case, only one stagnation point exists at $(x,y) = (-1,0)$. The bottom right image extends to larger distances to illustrate the formation of a tail with finite extension in the y-direction.}
       \label{fig:sketch}
       \end{figure}

The existence of a separatrix is a topological criterion for the determination of a \lq pause\rq. 
Separatrices form when a flow encounters an obstacle, which can be either a solid body or a counterflow (e.g. in form of a stellar wind or a planetary magnetosphere). At the location of the encounter, the driving forces of the flow and counterflow equalise and the flows have to decelerate,
leading to an isolated null-point of the flow. This null point is an intersection of a stagnation line (or curve) with another stagnation curve (the separatrix curve in 2D) or with a stagnation surface (the separatrix surface in 3D). In astrospheric physics, this separatrix surface is named astropause\footnote{In magnetic reconnection theory or magnetic topology where the 3D null point lies in its 2D subspace, the tangential subspace of the magnetopause is called a fan or fan surface \citep[see e.g. the reviews by][]{2005LRSP....2....7L, 2022LRSP...19....1P}.}.
A finite object (a so-called \lq Rankine body\rq) has at least two null-points\footnote{This statement is only true when in the case of a laminar potential flow around a cylinder, for instance, that is, a circle-like shaped solid object, there is no component of the flow with a circulation that reaches the critical value.}  \citep[e.g.,][]{1982MNRAS.202..735N, 1983A&A...118...57F} and an infinite object (so-called \lq Rankine half body\rq) at least one \citep{1961ApJ...134...20P}. Examples of 2D flow configurations are shown in Fig.\,\ref{fig:sketch}.

The physical reason for the presence of an isolated null point is the existence of an isolated stationary point $\bm{x}_0$, that is, a point at which $\bm{\nabla} 
p(\bm{x}_0)=\bm{0}$, such as an extremum of a scalar pressure $p$. 
Considering the momentum equation 
\begin{equation}
\rho(\bm{v\cdot\nabla})\bm{v} = -\bm{\nabla}p + \bm{\nabla\cdot\sigma}\, ,
\label{eq:euler0}
\end{equation}
we obtain for the case of ideal HD, for which the Navier-Stokes force term $\bm{\nabla\cdot\sigma} = \bm{0}$, at a stationary point of the pressure (provided that the density $\rho \neq 0$),
\begin{equation}
\bm{\nabla} p = -\rho(\bm{v\cdot\nabla})\bm{v}=\bm{0} \quad \Rightarrow \quad (\bm{v\cdot\nabla})\bm{v}=\bm{0} .
\label{eq:druck}
\end{equation}
To decompose the velocity vector $\bm{v}$, 
we used the Frenet-Serret formula $\frac{d \bm{e}_{t}}{dt} \equiv \kappa\bm{e}_{n}$, where $1/\kappa=\varrho$ is the radius of curvature, that is, $\kappa$ is the curvature, $\bm{e}_{n}$ is the unit vector in normal direction, and $\bm{e}_{t}=\bm{v}/{|\bm{v}|}$ is the tangential unit vector with $|\bm{v}|\equiv v_{t}$ \citep{Klingenberg}. 
The indices $t$ and $n$ describe the arc length and the normal direction, respectively.
With this, the right-hand side of the momentum equation, Eq. (\ref{eq:druck}), can be written in the following form:
\begin{eqnarray}
&& \left(v_{t}\bm{e}_{t}\bm{\cdot e}_{t}\,\frac{d}{dt} \right)\left(v_{t}\bm{e}_{t}\right)
= v_{t}^2\frac{d\bm{e}_{t}}{dt} + v_{t}\frac{d v_{t}}{dt}\bm{e}_{t} ={\mathbf 0}
\nonumber\\
\Rightarrow\,\, && v_{t}^2\kappa\bm{e}_{n}+v_{t}\frac{d v_{t}}{dt}\bm{e}_{t}={\mathbf 0} .
\end{eqnarray}
To fulfil this last form of the momentum equation, both terms on the left-hand side must vanish independently. Because for general streamlines $\bm{e}_{n}$ and $\bm{e}_{t}$ are regular and the curvature is finite and non-zero, this is only possible if $v_{t}\rightarrow 0$. 
Hence, at the stationary point of the pressure, there must be an isolated null point of the flow, that is, a stagnation point.

For the pure ideal HD case, we recognise from Eq.~\eqref{eq:druck} that
\begin{equation}
  \bm{\nabla} p(\bm{x}_0) = \bm{0} = - \rho(\bm{x}_0)\,\left(\bm{\nabla v}(\bm{x}_0)
    \right)\bm{v}(\bm{x}_0)\, \, ,
\end{equation} 
implying that $\bm{v}(\bm{x}_0)=\bm{0}$, as far as the mass density $\rho$
is non-zero and the rank of the velocity Jacobian matrix 
$\left(\bm{\nabla v}(\bm{x}_0) \right)$ is 3. This means that due to the linear independence, only the trivial solution $\bm{v}(\bm{x}_0)=\bm{0}$ is
possible\footnote{If the rank of the Jacobian is 2 or smaller, then there are solutions where not all components of the velocity necessarily have to vanish, sometimes so-called 2.5D solutions, such that there are no stagnation points at the position of the extremum.}. 
Thus, we conclude that a pressure extremum, together with a non-degenerate velocity field, implies a stagnation point\footnote{An analogy can be drawn from magnetohydrostatic (MHS)
theory, where the momentum equation is given by $\bm{\nabla}\Pi_{pB}=\left(\bm{\nabla B} \right) \bm{B}$.}.

This conclusion even holds for non-ideal HD, provided that the Navier-Stokes force term $\bm{\nabla\cdot\sigma}$ in Eq.~\eqref{eq:euler0} satisfies $\bm{\nabla\cdot\sigma} = f\bm{v}$,
where $f$ is either a constant or a function. An extension is the case of a regular bounded operator, probably written as a Jacobian matrix of $\bm{v}$, leading to a similar argument that we take up again below. 

\section{The (co-)location of null points}
\label{sect:coloc}

In the following, we show that under ideal constraints, that is, in analogy to the pure hydrodynamical case, the locations of the null points of the magnetic field, the velocity field, and the pressure gradient are identical. We mainly discuss the first-order momentum equations, consisting of the  centre of mass (CM) equation and the ideal or non-ideal Ohm's law as a relic of the electron momentum equation. The mass-continuity equation is considered only indirectly, where we assume that
if $\bm{v}$ is the CM velocity, then the stationary mass conservation $\bm{\nabla \cdot} (\rho \bm{v}) = 0$ is used.

\subsection{The basic MHD equations and the pressure extremum condition}
\label{sect:theorem}

We used normalised parameters, for which we introduced the normalisation constants $\hat{B}$, 
$\hat{\rho}$, $\hat{l}$, $\hat{p}$, and $\hat{v}_A$, where $\hat{v}_A=\hat{B}/\sqrt{\mu_{0}\hat{\rho}}$
is the normalised Alfv\'en velocity, and $\mu_{0}$ is the magnetic permeability in a classical vacuum. Let $\bm{v}$ be the plasma velocity normalised on 
$\hat{v}_A$, $\rho$ the mass density normalised on $\hat{\rho}$,
$\bm{j}=\bm{\nabla}\times\bm{B}$ the current density vector normalised on
$\hat{B}/(\mu_{0}\hat{l})$ with $\hat{l}$ as the characteristic length scale, and
$p$ the scalar plasma pressure normalised on
$\hat{p}=\hat{B}^2/\mu_{0}$.

With these definitions, we can formulate the set of normalised equations of stationary MHD. The momentum equation can then be written as
\begin{equation}
    \bm{\nabla} p =\left(\bm{\nabla}\times \bm{B}\right)\times \bm{B}\ -
    \rho \,  \left( \bm{v\cdot\nabla} \right) \bm{v}.
    \label{eq:momEqn}
\end{equation}
The pressure driving the stationary flow $\bm{v}$ in the case of ideal MHD is the total pressure $\Pi_{pB}$, consisting of the thermal and the magnetic pressure, 
\begin{equation}
   \Pi_{pB}:=p+\frac{\bm{B}^2}{2}.
   \label{eq:totPres}
\end{equation}
We investigated the conditions under which the co-location of the velocity null point and the magnetic null point can be identified a priori. The analysis is local and is formulated on a domain $\mathcal{D}$ in which the fields are assumed to be regular.
\begin{theorem}[Local co-location of non-degenerate nulls]
    Let $\bm{v},\bm{B} \in C^{1}(\mathcal{D};\mathbb{R}^{3})$ and let $ \bm{x}_{0} \in \mathcal{D} $ satisfy $ \bm{\nabla}\Pi_{pB}(\bm{x}_{0}) = \bm{0}$. If $\bm{v}(\bm{x}_{0})=\bm{0} $ and $\left( \bm{\nabla}\bm{B} \right)(\bm{x}_{0})$ is non-singular, then $\bm{B}(\bm{x}_{0})=\bm{0}$. Conversely, if $\bm{B}(\bm{x}_{0})=\bm{0} $, and  $\left( \bm{\nabla}\bm{v} \right)(\bm{x}_{0})$ is non-singular, then $\bm{v}(\bm{x}_{0})=\bm{0} $.
    \label{theorem1}
\end{theorem}
\begin{proof}
    The momentum equation, Eq.~\eqref{eq:momEqn}, can be rewritten in terms of the total pressure, Eq.~\eqref{eq:totPres} as
    \begin{equation}
        \bm{\nabla}\Pi_{pB} = \left(  \bm{B\cdot\nabla} \right)\bm{B} - \rho \left(  \bm{v\cdot\nabla} \right)\bm{v}.
        \label{eq:totPresEqn}
    \end{equation}
    Let $\bm{x}_{0} \in \mathcal{D}$ be a stationary point of $\Pi_{pB}$. At this point, Eq.~\eqref{eq:totPresEqn} reduces to
    \begin{equation}
        \left(  \bm{B\cdot\nabla} \right)\bm{B} - \rho \left(  \bm{v\cdot\nabla} \right)\bm{v} = 0.
        \label{eq:presExtrem}
    \end{equation}
    If $\left( \bm{\nabla}\bm{B} \right)$ and $\left( \bm{\nabla}\bm{v} \right)$ are the Jacobian matrices of the magnetic and the velocity field, respectively, Eq.~\eqref{eq:presExtrem} can be equivalently written as
    \begin{equation}
        \left( \bm{\nabla}\bm{B} \right)\bm{B} - \rho \left( \bm{\nabla}\bm{v} \right)\bm{v} = \bm{0}.
        \label{eq:presExtrem2}
    \end{equation}
    Thus, assuming that $\bm{v}(\bm{x}_{0}) = \bm{0}$, Eq.~\eqref{eq:presExtrem2} evaluated at $\bm{x}_{0}$ becomes
    \begin{equation}
        \left( \bm{\nabla}\bm{B}(\bm{x}_{0}) \right)\bm{B}(\bm{x}_{0}) = \bm{0}.
        \label{eq:Bnull}
    \end{equation}
    Since $\left( \bm{\nabla}\bm{B}(\bm{x}_{0}) \right)$ is non-singular, Eq.~\eqref{eq:Bnull} implies $ \bm{B}(\bm{x}_{0}) = \bm{0} $.
    Conversely, assuming $ \bm{B}(\bm{x}_{0}) = \bm{0} $ and evaluating Eq.~\eqref{eq:presExtrem2} at $\bm{x}_{0}$ gives
    \begin{equation}
        \rho(\bm{x}_{0}) \left( \bm{\nabla}\bm{v}(\bm{x}_{0}) \right)\bm{v}(\bm{x}_{0}) = \bm{0}.
        \label{eq:Vnull}
    \end{equation}
    As $\rho(\bm{x}_{0}) \neq 0$ and $\left( \bm{\nabla}\bm{v}(\bm{x}_{0}) \right)$ is non-singular, Eq.~\eqref{eq:Vnull} implies $\bm{v}(\bm{x}_{0}) = \bm{0}.$ This proves both implications. 
\end{proof}

A special case of the aforementioned is when the flow can be treated as incompressible, that is, $\bm{\nabla\cdot v} = 0$. Then, Eq.~\eqref{eq:presExtrem}
can be rewritten by substituting $\bm{w}:=\sqrt{\rho}\, \bm{v}$, 
  \begin{eqnarray}
    \left( \bm{B\cdot\nabla} \right)  \bm{B} -  \rho\,  \left( \bm{v \cdot\nabla} \right)
\bm{v} \equiv\left(\bm{\nabla B} \right) \bm{B} - \left(\bm{\nabla w}
    \right)\bm{w} = \bm{0}  \, . \label{euler4}
\end{eqnarray}
If the rank of the two Jacobians is three, the magnetic and flow null must be co-located, which emphasises the symmetry between the flow and magnetic field topology.

The result also shows that the non-degeneracy condition is essential. If the rank of the relevant Jacobian is smaller than three, the corresponding linear system can possess non-trivial solutions, and the vanishing of the total-pressure gradient does not necessarily force the other field to vanish. In this case, the magnetic and velocity null points need not be co-located. If in the entire local domain $\mathcal{D}$ only one extremum of the total pressure $\Pi_{pB}$ exists along with a null point of either the magnetic or the velocity field, then there exists only one null point of both vector fields, and the nulls are trivially identical. Our results for the conditions of co-location for the different cases of ideal fluids are summarised in Table~\ref{tab:colocation}.

The result is formulated in the domain $\mathcal{D}$, where we assumed that the fields are regular and differentiable. Consequently, they can be expanded linearly about some point $\bm{x}_{0}$. Taking $\bm{x}_{0} = \bm{0}$, we have the representations
\begin{equation}
    \bm{v} = \left( \bm{\nabla}\bm{v} \right)\bm{r} + R_{v}(\bm{r}) \quad \text{and} \quad \bm{B} = \left( \bm{\nabla}\bm{B} \right)\bm{r} + R_{B}(\bm{r}),
\end{equation}
where $R_{v}$ and $R_{B}$ denote the corresponding remainder functions associated with higher-order contributions. Considering co-location, it is possible to relate the topology of the magnetic field and flow dynamics through a parametric analysis of the eigenstructure of the associated fields \citep{2000PhPl....7.3542T,2026saddal}.

{\small
\begin{table}
\caption{Conditions for ideal fluids under which the null points of the flow and magnetic field co-locate at a point $x_{0}$.}
\label{tab:colocation}
\centering
\begin{tabular}{llll}
\hline \hline
Case & Given & Rank condition & Conclusion \\
\hline
HD    & $\bm{\nabla}p = \bm{0}$         & $\mathrm{rank}\left(\bm{\nabla v}\right)=3$ & $\bm{v}=\bm{0}$ \\
       & $\bm{\nabla}p = \bm{0}$        & $\mathrm{rank}\left(\bm{\nabla v}\right)<3$ & n.n. \\
\hline
MHS        & $\bm{\nabla}\Pi_{pB} = \bm{0}$ & $\mathrm{rank}\left(\bm{\nabla B}\right)=3$ & $\bm{B}=\bm{0}$ \\
           & $\bm{\nabla}\Pi_{pB} = \bm{0}$ & $\mathrm{rank}\left(\bm{\nabla B}\right)<3$ & n.n. \\
\hline
MHD   & $\bm{\nabla}\Pi_{pB} = \bm{0}$, $\bm{v}=\bm{0}$ & $\mathrm{rank}\left(\bm{\nabla B}\right)=3$ & $\bm{B}=\bm{0}$ \\
           & $\bm{\nabla}\Pi_{pB} = \bm{0}$, $\bm{v}=\bm{0}$ & $\mathrm{rank}\left(\bm{\nabla B}\right)<3$ & n.n. \\
           & $\bm{\nabla}\Pi_{pB} = \bm{0}$, $\bm{B}=\bm{0}$ & $\mathrm{rank}\left(\bm{\nabla v}\right)=3$ & $\bm{v}=\bm{0}$ \\
           & $\bm{\nabla}\Pi_{pB} = \bm{0}$, $\bm{B}=\bm{0}$ & $\mathrm{rank}\left(\bm{\nabla v}\right)<3$ & n.n. \\
\hline
\end{tabular}
\tablefoot{Entries termed n.n. (not necessarily) refer to cases where the corresponding null does not need to exist at that location.} 
\end{table}
}

\subsection{The co-location in the case of resistive MHD}
\label{sect:resistive}

We so far neglected the second momentum equation, namely the ideal, or
here, the resistive Ohm's law $\bm{E} +  \bm{v}\times \bm{B}= \eta \bm{j}$,
where $\bm{E}$ and $\bm{j}$ are the electric field and electric current density, respectively, 
and $\eta$ is the resistivity.
Combined with the (stationary) Faraday's induction equation $\bm{\nabla}\times
\bm{E}= - \partial \bm{B} /\partial t = \bm{0} $ and using
vector analysis identities, we obtain
\begin{eqnarray}
\bm{E} +  \bm{v}\times \bm{B}= \eta \bm{j}\,\,& \Rightarrow & \,\, \bm{\nabla}\times\left( \bm{v}
\times \bm{B} \right)= \bm{\nabla}\times (\eta \bm{j})
\label{ohm1}
.\end{eqnarray}
This equation is termed the induction equation of resistive MHD.
 \citet{2008ASTRA...4....7N} proved that for $\eta=0$, the clear
conclusion can be drawn directly from Eq.~\eqref{ohm1} that a stagnation
point must be a magnetic null point and vice versa because of the structure of
the ideal Ohm's law and the condition that all fields are three-dimensional. 
When we assume that $\eta \bm{j}$ can be written as a gradient of a potential $\phi$,
\begin{equation}
    \eta \bm{j} = -\bm{\nabla}\phi,
    \label{eq:potential}
\end{equation}
the system is quasi-ideal because the curl of the 
total electric field vanishes identically \citep[see, e.g.,][]{1998JPlPh..59..303T, 2000JPlPh..64..601T,
2000PhPl....7.3542T, 2014A&A...569A..44N}. Then, we have exactly the same case
as in \citet{2008ASTRA...4....7N}, namely the automatic co-location of the velocity and magnetic field nulls. 
That the resistive term can be written completely as a gradient can be regarded as too specific, 
although it is physically possible. It should be considered, however, that for 
the co-location constraint, it would also suffice if $\eta \bm{j}$ can be written 
in the form of Eq.~\eqref{eq:potential} in the immediate vicinity of the stagnation point or magnetic null point. 

\subsection{Dislocation of the null points in resistive MHD}
\label{sect:disloc}

We next estimated the typical separation between the magnetic null and the velocity null when they are dislocated in the linear regime. The induction equation can be written as
\begin{equation}
    \nabla\times(\bm{v}\times\bm{B})
    =
    (\bm{B}\cdot\nabla)\bm{v}
    -(\bm{v}\cdot\nabla)\bm{B}
    -(\nabla\cdot\bm{v})\bm{B}
    =
    \nabla\times\bm{N}.
\end{equation}
When we define $\mathcal{M}:=\nabla\bm{v}-(\nabla\cdot\bm{v})\mathcal{E}$ with the unit matrix $\mathcal{E}$, the evaluation at the velocity null $\bm{x}_s$, where $\bm{v}(\bm{x}_s)=0$, gives
\begin{equation}
    \left.\mathcal{M}\bm{B}\right|_{\bm{x}_s}
    =
    \left.\nabla\times\bm{N}\right|_{\bm{x}_s}.
\end{equation}
For a non-degenerate stagnation point, $\mathcal{M}$ is invertible, and hence, $\bm{B}(\bm{x}_s)=\mathcal{M}^{-1}(\nabla\times\bm{N})|_{\bm{x}_s}$. Expanding the magnetic field linearly about the stagnation point and evaluating at the magnetic null $\bm{x}_B$ gives $\bm{0}\simeq\bm{B}(\bm{x}_s)+(\nabla\bm{B})_{\bm{x}_s}\bm{d}$, where $\bm{d}:=\bm{x}_B-\bm{x}_s$. Thus,
\begin{equation}
    \bm{d}
    \simeq
    -(\nabla\bm{B})^{-1}
    \mathcal{M}^{-1}
    \left.(\nabla\times\bm{N})\right|_{\bm{x}_s}.
\end{equation}
For $\bm{N}\sim\eta\bm{j}$ and constant $\eta$, $\nabla\times\bm{N}=-\eta\nabla^2\bm{B}$ so that $\lvert\nabla\times\bm{N}\rvert\sim\eta B/L^2$, while $\lvert\mathcal{M}\rvert\sim v/L$ and $\lvert\nabla\bm{B}\rvert\sim B/L$. Consequently, the characteristic null-point separation is
\begin{equation}
    |\bm{d}|\sim\frac{\eta}{v},
    \qquad
    \frac{|\bm{d}|}{L}\sim\frac{\eta}{vL}
    =\frac{1}{R_m},
\end{equation}
where $R_m=vL/\eta$ is the magnetic Reynolds number. Thus, for systems with $R_m\gg1$, such as the heliosphere, the displacement of the magnetic and velocity nulls is small compared to the characteristic macroscopic length scale $|\bm{d}|\ll L$.

\subsection{A brief discussion of Hall MHD}
\label{sect:Hall}

Instead of a pure resistive term in Eq.~\eqref{ohm1}, we assume in the following a non-ideal term that includes the Hall term and the electron pressure ($p_e$), neglecting effects of anisotropy. Because of the quasi-neutrality, we assumed that $n_i \simeq n_e \simeq n$, $\bm{v_i} \simeq \bm{v}$ and took in addition the definition equation for the total current of ions and electrons. Then the total electric fields seen by the electrons and ions are given by
\begin{eqnarray}
\bm{E} +  \bm{v}_e\times \bm{B} & = & - \frac{1}{n e}\bm{\nabla} p_{e}  \quad \land \quad \bm{v}_e = \bm{v} - \frac{\bm{j}}{n e}  \label{tot_E_e}\\
\Leftrightarrow  \qquad  \bm{E} + \bm{v}\times \bm{B} & = & \frac{1}{n e} \left( \bm{j}\times \bm{B} - \bm{\nabla} p_{e} \right)\,.  
\label{tot_E_i}
\end{eqnarray}
The inclusion of the Hall term leads to a symmetry breaking so that co-location is not fulfilled automatically. Nevertheless, it can take place under certain circumstances.
When we consider that $p_e$ is a function of $n$ or vice versa, taking the curl of 
Eq.~(\ref{tot_E_e}) delivers the same co-location statement for 3D fields as in the ideal MHD case 
\citep{2008ASTRA...4....7N}, but only for the electron velocity and the magnetic field. A discussion of further possibilities for co-location and dislocation with respect to the velocity field $\bm{v}$ is presented in Appendix~\ref{sect:app}.

\subsection{The general non-ideal case}
\label{sect:general}

A generalisation of Ohm's law can be carried out by replacing the resistive term in Eq.~\eqref{ohm1} by a general non-ideal term $\bm{N}$. Applying the Helmholtz theorem, we can write this term as the sum of the gradient of a potential $\phi_{{\rm e,total}} $ and the curl of a vector potential $\bm{W}$ such that $\bm{N}:=-\bm{\nabla}
\phi_{{\rm e,total}} + \bm{\nabla}\times \bm{W}$. Hence, the induction equation can then be formulated as
\begin{equation}
   \left(\bm{\nabla v} \right)   \bm{B} - \left(\bm{\nabla\cdot v}
    \right)  \bm{B} - \left(\bm{\nabla B}
  \right)  \bm{v} =
  -\Delta  \bm{W} \, ,
  \label{ohm3}
  \end{equation}
as far as $\bm{W}$ satisfies the Coulomb gauge. We note that only the solenoidal part of $\bm{N}$ enters Eq.~(\ref{ohm3}), while the gradient contribution $-\bm{\nabla}\phi_{{\rm e,total}}$ is annihilated by the curl and is therefore invisible to the stationary Faraday's law. Consequently, the case $\Delta\bm{W}=\bm{0}$ recovers the quasi-ideal case discussed in Sect.~\ref{sect:resistive}.

The condition $\Delta\bm{W}=\bm{0}$ is very restrictive and implies that each Cartesian component of $\bm{W}$ is harmonic.
Consequently, $\bm{W}$ satisfies the mean-value property: for any ball $B_{r}(\bm{x}_0)$ contained in the domain,
\begin{equation}
\bm{W}(\bm{x}_0) = \frac{1}{|\partial B_r|}
\int_{\partial B_r} \bm{W}(\bm{x}) dS\, .
\end{equation}
Thus, the value of the field at $\bm{x}_0$ is equal to the surface average of the field over any sphere centred
at $\bm{x}_0$. In this sense, a vanishing Laplacian characterises a local balance of the field with respect to its surroundings: there is no net diffusive imbalance at $\bm{x}_0$. Importantly, this condition does not imply that $\bm{W}$ is locally constant because $\bm{\nabla W}$ might remain non-zero. Rather, the local value of $\bm{W}$ is constrained by the surrounding field through the mean-value property.

To obtain a consequence similar to that stated in theorem 1, it is sufficient to consider $\Delta \bm{W}(\bm{x}_0) = \bm{0}$ for a given $\bm{x}_0$ at which the analysis is continued. In this case, the diffusive contribution vanishes locally. This condition indicates that there is no net local diffusive imbalance of the field $\bm{W}$ at $\bm{x}_0$.
However, this pointwise condition should not be confused with the harmonicity of $\bm{W}$ in a neighbourhood of $\bm{x}_0$:
the mean-value property and the corresponding maximum principle require $\Delta \bm{W} = \bm{0}$ throughout an open neighbourhood (or domain), rather than only at a single point.

We studied the class of non-ideal terms for which the solenoidal part is generated linearly by the two vector fields $\bm{v}$ and $\bm{B}$ themselves,
\begin{equation}
  -\Delta \bm{W} \;=\; \mathcal{F}(\bm{r})\,\bm{v} \;+\; \mathcal{G}(\bm{r})\,\bm{B} \, ,
  \label{eq:FGansatz}
\end{equation}
where $\mathcal{F}(\bm{r})$ and $\mathcal{G}(\bm{r})$ are $3\times 3$ matrices whose components are functions of the spatial coordinates alone. Because $-\Delta\bm{W} = \bm{\nabla}\times\bm{N}$ is by
construction a curl, admissible matrix pairs $\left(\mathcal{F},\mathcal{G}\right)$ 
are restricted by the solenoidality condition
\begin{equation}
  \bm{\nabla\cdot}\left( \mathcal{F}(\bm{r})\,\bm{v} + \mathcal{G}(\bm{r})\,\bm{B} \right) = 0 \, ,
  \label{eq:FGsolenoidal}
\end{equation}
which must hold throughout the domain under consideration.

Substitution of the ansatz, Eq.~(\ref{eq:FGansatz}), into the induction equation, Eq.~(\ref{ohm3}), and collection of the terms proportional to $\bm{B}$ and $\bm{v}$ gives
\begin{equation}
  \left( \left(\bm{\nabla v}\right) - \left(\bm{\nabla\cdot v}\right)\mathcal{E}
         - \mathcal{G} \right) \bm{B}
  \;-\;
  \left( \left(\bm{\nabla B}\right) + \mathcal{F} \right) \bm{v}
  \;=\; \bm{0} \, .
  \label{eq:FGmain}
\end{equation}
The induction equation hence adopts a form that is structurally identical to Eq.~(\ref{eq:presExtrem2}). The non-ideal contributions therefore act as additive corrections to the velocity and magnetic-field Jacobians, respectively, so the argument used for the ideal case carries over unchanged.
\begin{corollary}
\label{cor:FG}
Consider the class of non-ideal terms given by Eq.~(\ref{eq:FGansatz}), and let
$\bm{v},\bm{B} \in C^{1}(\mathcal{D};\mathbb{R}^{3})$ with the components of $\mathcal{F}$ and $\mathcal{G}$ bounded and continuous on $\mathcal{D}$. If at a point $\bm{x}_{0} \in \mathcal{D}$
\begin{enumerate}
  \item[(i)] $\bm{B}(\bm{x}_{0}) = \bm{0}$ and
        $\mathrm{rank}\left[ \left(\bm{\nabla B}\right) + \mathcal{F} \right](\bm{x}_{0}) = 3$,
        then $\bm{v}(\bm{x}_{0}) = \bm{0}$;
  \item[(ii)] $\bm{v}(\bm{x}_{0}) = \bm{0}$ and
        $\mathrm{rank}\left[ \left(\bm{\nabla v}\right) - \left(\bm{\nabla\cdot v}\right)\mathcal{E}
        - \mathcal{G} \right](\bm{x}_{0}) = 3$,
        then $\bm{B}(\bm{x}_{0}) = \bm{0}$. 
\end{enumerate}
\end{corollary}
\noindent We note that $\mathcal{F}$ and $\mathcal{G}$ need not be non-singular, since only the combined matrices in Eq.~\eqref{eq:FGmain} are required to have full rank. In the following, we give several physical examples that fall within this class.
\begin{enumerate}
    \item \emph{Quasi-ideal dynamics.} If the non-idealness is a pure gradient, $\bm{N}=-\bm{\nabla}\phi_{\rm e,total}$, its curl vanishes identically and Eq.~(\ref{eq:FGansatz}) holds with $\mathcal{F}=\mathcal{G}=0$.
    \smallskip
    \item \emph{Constant resistivity in the vicinity of a linear magnetic null.} For $\bm{N}=\eta\bm{j}$ with constant $\eta$, we have $\bm{\nabla}\times\bm{N}=-\eta\,\Delta\bm{B}$. For a linear field, $\Delta\bm{B}=0$, giving $\mathcal{F}=\mathcal{G}=0$.
    \smallskip
    \item \emph{Linear force-free fields.} If $\bm{\nabla\times\bm{B}}=\alpha\bm{B}$ with constant $\alpha$ and $\eta$, then $\bm{\nabla}\times\bm{N}=\eta\alpha^{2}\bm{B}$, which has the form of Eq.~(\ref{eq:FGansatz}) with $\mathcal{G}=\eta\alpha^{2}\mathcal{E}$ and $\mathcal{F}=0$.
    \smallskip
    \item \emph{Non-constant resistivity and non-linear force-free fields.} For $\bm{\nabla\times\bm{B}}=\alpha(\bm{r})\bm{B}$, we can write $\bm{\nabla} \times (\eta \bm{j}) = \alpha \bm{\nabla} \eta \times \bm{B} +  \eta\left(\bm{\nabla}\alpha \times \bm{B} + \alpha^{2}\bm{B} \right)$, which is identical to
    \begin{equation}
        \bm{\nabla} \times (\eta \bm{j}) = \left( \alpha \mathcal{R} + \eta \mathcal{A} + \eta\alpha^{2}\mathcal{E}  \right) \bm{B},
        \label{eq:forceFree}
    \end{equation}
    where
    \begin{equation}
        \mathcal{R} =
        \begin{pmatrix}
            0 & - \frac{\partial \eta}{\partial z} & \frac{\partial \eta}{\partial y} \\
            \frac{\partial \eta}{\partial z} & 0 & - \frac{\partial \eta}{\partial x} \\
            -\frac{\partial \eta}{\partial y} & \frac{\partial \eta}{\partial x} & 0 
        \end{pmatrix},
        \quad \mathcal{A} =
        \begin{pmatrix}
            0 & - \frac{\partial \alpha}{\partial z} & \frac{\partial \alpha}{\partial y} \\
            \frac{\partial \alpha}{\partial z} & 0 & - \frac{\partial \alpha}{\partial x} \\
            -\frac{\partial \alpha}{\partial y} & \frac{\partial \alpha}{\partial x} & 0 
        \end{pmatrix}.
    \end{equation} 
    Equation~\eqref{eq:forceFree} has the same structure as Eq.~\eqref{eq:FGansatz}, with $\mathcal{F} = 0$ and $\mathcal{G} =  \alpha \mathcal{R} + \eta \mathcal{A} + \eta\alpha^{2}\mathcal{E}.$
    \smallskip
    \item \emph{Resistivity stratified along the current.} For a spatially varying resistivity, $\bm{\nabla}\times(\eta\bm{j}) = -\eta\Delta\bm{B} + \bm{\nabla}\eta\times\bm{j}$. If the resistivity varies only along the current direction, $\bm{\nabla}\eta\parallel\bm{j}$, the additional contribution
vanishes, and the situation reduces to case~(2).
\end{enumerate}

Two important cases fall outside this class because they introduce a preferred direction that is not constructed from $\bm{v}$ and $\bm{B}$. (i) For a resistivity varying transverse to the current, the term $\bm{\nabla}\eta\times\bm{j}$ is generally transverse to both fields. (ii) The Hall term is analogous, with $\bm{\nabla}n$ introducing an additional direction and yielding $\bm{\nabla}\times\bm{N} = -\left(\bm{\nabla B}\right)\bm{j}/ne$ at a magnetic null. This agrees with Sect.~\ref{sect:Hall}. In Hall MHD, it is the electron stagnation point, rather than the CM stagnation point, that co-locates with the magnetic null.

\section{Discussion and conclusions}
\label{sect:disc}

We performed the first comprehensive analysis of the physical conditions that need to be met so that the magnetic null point and the stagnation point co-locate. Our motivation for this study was based on the question whether the astropause
is represented by the magnetopause or by the hydropause, or if they are identical at least within a small domain around the apex.
If the nulls of the magnetic and flow fields were offset, the magnetic separatrix and the velocity separatrix might not even coincide locally, 
such that several boundaries between the interstellar stream and the stellar wind can exist.
Our analysis was performed in the frameworks of ideal HD and of ideal, resistive, non-ideal, and Hall MHD.

We wrote the ideal HD momentum equation, Eq. (\ref{eq:druck}), in a form that reflects 
the locally quasi-linear mapping character of the (albeit in reality) non-linear
differential operator. This type of tool of re-writing the force
operators of the fluid equation enabled us to argue that at an
extremum of the pressure, the velocity must vanish at this location if
the velocity field is three-dimensional there.

A generalisation for the case of MHD can be achieved when we consider an extremum of the total pressure that is co-located with a null point of one of the two fields. As we showed, if under these conditions the Jacobian of the respective other field is
non-singular at this location, then that field must vanish there as well, meaning that stagnation points must be magnetic null points, and vice versa. 
An extremum of the total pressure alone is insufficient to postulate the existence of any null point. At the location of such an extremum, we calculated the velocity field from a given magnetic field and vice versa, provided that the corresponding Jacobian is prescribed and non-singular, but neither 
of the fields needs to
vanish there. Conversely, one of the vector fields may possess a null at a location where neither the other field nor the pressure gradient vanishes. In that case, the pressure gradient is balanced by a single remaining force, the inertial force at a magnetic null and the magnetic tension at a stagnation point, since the other term vanishes together with the field.
However, when the second momentum equation is invoked, this freedom is removed. For an ideal Ohm's law, the stationary Faraday's law yields
$\bm{\nabla}\times(\bm{v}\times\bm{B})=\bm{0}$, and the two nulls then coincide, provided the corresponding Jacobian is non-singular. More generally, they coincide if the curl of the non-ideal term vanishes at the null.

When we compare the results from our exact analysis with those from MHD simulations, it is surprising that these mostly find an offset between the magnetic and the flow null
despite the fact that these simulations were performed with (in principle) an ideal Ohm's law \citep[e.g.,][]{2015ApJ...805..173R, 2020A&A...633L..12I, 2024ApJ...961..244F}.
The reasons for the discrepancies between the numerical and exact analytical results should be investigated in more detail. However, this is difficult because not all information on the input for the simulations can be transparently communicated. If the simulations use ideal theory, as is claimed by the authors, then the results should also clearly reflect the ideal stationary behaviour.

When we consider resistive (or non-ideal in general) MHD, co-location of the nulls can occur for specific resistivities and magnetic  field configurations. We showed that if dislocation occurs, it is of the order of $L/R_{m}$ and is typically small for the physical conditions in the heliosphere.

In the context of Hall MHD, the Hall term leads to a break in the symmetry because electrons and ions are treated separately. If co-location takes place, it is typically for the stagnation point of the electrons with the magnetic null point, while the stagnation point of the ions can be offset. While we did not calculate the offset distances directly, we estimated the influence of the Hall effect on them by
comparing the ion inertial scale with the typical length scale of the change in the magnetic field for different environments, such as the heliopause and the Earth magnetopause. We found that for the heliopause, the effect of the Hall term is very small and concluded that the smaller the Hall effect, the smaller the offset. The offset found in the simulations by \citet{2015ApJ...800L..28O} seems to be considerably larger than what we would expect from the Hall term alone.

When we release the condition of 3D null points, meaning that the rank of the Jacobian matrix of the magnetic field is smaller than 3 at a 3D velocity null, then the co-existence of a magnetic null
is not necessarily fulfilled. Instead, this configuration resembles that of a magnetic flipping or quasi-separatrix layer \citep{1992JGR....97.1521P, 1995JGR...10023443P, 2025PhPl...32i3904R} rather than to a proper magnetic separatrix. The problem is that such configurations do not enable us to distinguish between 
\lq outer\rq~and \lq inner\rq~magnetic field, that is, they do not allow for a topological separation.
The analogy is that the magnetic field has a 3D null, and there exists only a HD quasi-separatrix 
layer that does not enable us to distinguish between \lq outer\rq~and \lq inner\rq~medium.
   
The MHD equations, investigated from a general perspective in the vicinity of null points of vector fields, show that the occurrence of a co-location of null points of different vector fields (here, the magnetic and the velocity field) is natural when only one of the momentum equations is of ideal character. When the two momentum equations include non-ideal terms, the null points can be displaced. An exception of this is the case where the non-ideal terms present a linear behaviour. Then, the co-location can be realised. 
In general, we might expect that the use of classical multi-fluid terms, such as inertial terms and the Hall term, prevents the representation as given in Eq.~\eqref{eq:FGansatz}, since these terms introduce a preferred direction that is not constructed from $\bm{v}$ and $\bm{B}$ alone.
In particular, the Hall term can result in co-location only for the stagnation point of the electrons (rather than the CM) and the magnetic null point, as shown in Appendix~\ref{sect:app}.

For a generalised Ohm's law, it is sufficient for the co-location of the null points that at the magnetic null point, the curl of the non-ideal term has to vanish locally, whereas 
in the quasi-ideal dynamics, the global curl of the non-idealness vanishes automatically, as it is represented by a gradient (e.g. electron pressure in Hall MHD). An implication is that if all the external forces can be completely written as gradients, the co-location of the null points of the flow and magnetic field are 
guaranteed throughout the domain wherever the corresponding Jacobian is non-singular. To achieve co-location when the external forces are anisotropic, meaning that they additionally have a component resulting from the curl of a vector potential, the curl of the external forces has at least to vanish at the null.

Even if the null points of flow and magnetic field are co-located, one open issue is whether the 
tangential magnetopause (magnetic fan) is identical to the tangential astropause 
(hydrodynamic fan). This requires further investigation. When the two fans are identical locally, 
the second issue is whether the two separatrices are also identical globally, that is, whether magnetopause 
and astropause are identical. Furthermore, the question arises whether only field-aligned flows 
in the vicinity of the astropause guarantee the existence of an at least local identity of 
the two separatrices, that is, whether the two separatrix surfaces are tangential.

While it is always robust and desirable that the numerical and theoretical results can be 
compared with observations, the context we considered does not allow for such an analysis in a 
fully general manner. Future investigations, such as performing a linear instability 
analysis, might therefore help us to identify why configurations with a non-co-location are observed in 
simulations and observations, for example in quasi-2D planetary magnetopause reconnection processes
\citep[e.g.][]{2007PhPl...14j2114C, 2008JGRA..113.6210P, 2016Sci...352.2939B}.

\begin{acknowledgements}
This research made use of the NASA Astrophysics Data System (ADS).
D.H.N., K.S.S., and M.K. acknowledge financial support from the Czech Science Foundation 
(GA \v{C}R, grant number 25-17532S). The Astronomical Institute of the Czech Academy of Sciences 
is supported by the project RVO:67985815. This project is co-funded by the European Union (Project 
101183150 - OCEANS). 
\end{acknowledgements}

\bibliography{nullpoint}
\bibliographystyle{aa}

\onecolumn

\begin{appendix}
\nolinenumbers

\section{Co-location and dislocation in the framework of stationary flows in Hall MHD}
\label{sect:app}

We derive conditions under which magnetic null and stagnation point in stationary Hall MHD are co-located. We start from $\bm{E} = - \bm{\nabla} \phi_{e}$, consider the definition equation for the current $\bm{v}_e = \bm{v} - \frac{\bm{j}}{n e}$, and take into account the relation for the total electric field seen by the electrons and ions
\begin{equation}
\bm{E} +  \bm{v}_e\times \bm{B}  = - \frac{1}{n e}\bm{\nabla} p_{e} \qquad \Leftrightarrow  \qquad  \bm{E} + \bm{v}\times \bm{B}  =  \frac{1}{n e} \left( \bm{j}\times \bm{B} - \bm{\nabla} p_{e} \right)\,.  
\label{app:tot_E}
\end{equation}
Considering that $p_e$ is a function of $n$ or vice versa and taking the curl of the right-hand side equation of Eq.~\eqref{app:tot_E} leads to
\begin{equation}
\bm{\nabla}\times\left( \bm{v}
\times \bm{B} \right)= -\frac{1}{n^{2}e}\left[ 
\bm{\nabla} n \times (\bm{j} \times \bm{B}) \right] + \frac{1}{ne} \left[ (\bm{B} \cdot \bm{\nabla}) \bm{j} + (\bm{\nabla} \cdot \bm{B}) \bm{j} - (\bm{j} \cdot \bm{\nabla}) \bm{B} - (\bm{\nabla} \cdot \bm{j}) \bm{B} \right] 
\label{app:curl}
\end{equation}
Equation~\eqref{app:curl} can be reformulated in the following way,
\begin{equation}
(\bm{\nabla} \bm{v})\,\bm{B} - (\bm{\nabla} \bm{B})\,\bm{v} - (\bm{\nabla} \cdot \bm{v})\,\bm{B} = -\frac{1}{n^{2}e}\left[ 
(\bm{B} \cdot \bm{\nabla} n)\,\bm{j} - (\bm{j} \cdot \bm{\nabla} n)\,\bm{B}) \right] + \frac{1}{ne} \left[ (\bm{\nabla}\bm{j})\,\bm{B}  - (\bm{\nabla}\,\bm{B}) \bm{j} \right]\,. 
\label{app:reform}
\end{equation}
To investigate the conditions for possible co-location, we discuss the several cases separately.

\subsection{Stagnation point}

If we concentrate on the stagnation point and consider  
$\bm{v} = \bm{0}$, rank$(\bm{\nabla} \bm{v}) = 3$, and assume that at the stagnation point the density $\rho \propto n$ can be considered as constant, which means that
$\bm{\nabla} n = \bm{0} \Longrightarrow \bm{\nabla} \cdot \bm{v} = 0$,
the form of the induction equation, Eq.~\eqref{app:reform}, is given by 
\begin{equation}
(\bm{\nabla} \bm{v})\,\bm{B} =  \frac{1}{ne} \left[ (\bm{\nabla}\bm{j})\,\bm{B}  - (\bm{\nabla}\,\bm{B}) \bm{j} \right] 
\qquad \Leftrightarrow \qquad 
\left[ (\bm{\nabla} \bm{v}) -  \frac{1}{ne}  (\bm{\nabla}\bm{j}) \right]\,\bm{B} = - \frac{1}{ne} (\bm{\nabla}\,\bm{B}) \bm{j} 
\qquad \Leftrightarrow \qquad 
{\cal M} \,\bm{B} = - \frac{1}{ne} (\bm{\nabla}\,\bm{B}) \bm{j}\,.
\label{app:coloc_stag}
\end{equation}
For co-location to take place, Eq.~\eqref{app:coloc_stag} must be transformed into a homogeneous equation for $\bm{B}$. This can be achieved in several ways. 
In case the matrix ${\cal M}$ retains a rank of three, $\bm{B} = \bm{0}$ occurs if 
(i) $\bm{j} = \bm{0}$, meaning that there is a potential magnetic null \citep{1996PhPl....3..759P}, 
(ii)
$\bm{B}$ is a force-free field ($\bm{j}= \alpha\bm{B}$, where $\alpha$ is the force-free parameter) and rank$(\bm{B}) = 3$,
or 
(iii) rank$(\bm{\nabla}\,\bm{B}) = 2$ and $\bm{j}$ accordingly but $\bm{j}\neq \bm{0}$, implying that $(\bm{\nabla}\,\bm{B}) \bm{j} = \bm{0}$.

\subsection{Magnetic null point}

If we concentrate on the magnetic null point and consider 
$\bm{B} = \bm{0}$ and rank$(\bm{\nabla} \bm{B}) = 3$, the form of the induction equation, Eq.~\eqref{app:reform}, is given by
\begin{equation}
(\bm{\nabla} \bm{B})\,\bm{v} = \frac{1}{ne}(\bm{\nabla}\,\bm{B}) \bm{j} 
\qquad \Leftrightarrow \qquad  
(\bm{\nabla} \bm{B}) \left( \bm{v} - \frac{1}{ne} \bm{j} \right) = \bm{0} 
\qquad \Leftrightarrow \qquad 
(\bm{\nabla} \bm{B})\,\bm{v}_{e} = \bm{0} \,. 
\label{app_null}
\end{equation}
From the right-hand side equation of Eq.~\eqref{app_null} follows immediately that the electron velocity must be $\bm{v}_{e} = \bm{0}$. Since $\bm{v} \simeq \bm{v}_{i}$, the stagnation point of the ions should be in the vicinity of the one of the electrons.

To quantify the significance of the Hall term, we inspect two examples, the Earth magnetopause and the heliopause. For both, we estimate typical values for the ion inertial scale, defined by 
$c/\omega_{pi}$ where $c$ is the speed of light and $\omega_{pi}$ is the plasma frequency of the ions given by $\omega_{pi} = \sqrt{ne^{2}/(m_{i} \varepsilon_{0})}$ with the electric constant $\varepsilon_{0}$, and compare it with macroscopic length scales on which the magnetic field changes significantly.

Using a typical value of $n \simeq 2\cdot 10^{7}$m$^{-3}$ for the density of the Earth magnetopause of \citep[e.g.,][]{1994JGR....99..121P} results in a scale length of the ion inertia of $\sim 50$\,km. This length is about $10\%$ of the value of the thickness $L \simeq 500-700$\,km of the Earth magnetopause around the magnetic equator \citep{1982JGR....87.2108B, 2017Sci...356..960R}. This means that the Hall term cannot be neglected, and the stagnation point of the ions can be expected to have a non-negligible offset from the stagnation point of the electrons and its co-located magnetic null point.

The ratio of the length scales changes significantly when we turn to the conditions at the heliopause. 
The macroscopic length scale, on which the magnetic field varies across the heliopause, is on the order of $\sim 5\cdot 10^{6}$\,km \citep{2013Sci...341..147B}.
For typical values of the particle density of $\sim 2\times 10^{3}$\,m in the outer heliosphere \citep{2019NatAs...3.1024G} the scale length of the ion inertia is 
$\sim 5\cdot 10^{3}$\,km and hence comprises only $0.1\%$ of the macroscopic length scale. One may expect that the smaller the Hall effect the smaller the offset.    

\subsection{Quasi-static approach}

In the quasi-static approach, the term $\bm{j} \times \bm{B}$ can be written as $\bm{\nabla} p$, such that the curl of the right-hand side of Eq.~\eqref{app:tot_E} is zero, 
if the pressure $p-p_{e}$ depends on $n$
explicitly. A more general case would be that there exists a generalized pressure $\tilde{p}$ such that 
$ne \bm{\nabla} \tilde{p} + \bm{\nabla} p_{e} = \bm{j} \times \bm{B}$. This is an analogous case to quasi-ideal MHD and implies co-location.

\end{appendix}

\end{document}